\documentclass[a4paper,UKenglish,cleveref, autoref, thm-restate,authorcolumns]{lipics-v2021}

\def\MdN{\ensuremath{\mathbb{N}}}
\newcommand{\Oh}[1]{\mathcal{O}\!\left( #1\right)}
\newcommand{\Otilde}[1]{\tilde{\mathcal{O}}\!\left( #1\right)}
\newcommand{\Is}{:=}
\newcommand{\etal}{et~al.\ }
\newcommand{\ie}{i.e.\ }
\newcommand{\CC}{C\texttt{++}}
\newcommand{\papertitle}{Practical Fully Dynamic Minimum Cut Algorithms}

\newif\ifVLDB
\VLDBfalse

\newif\ifEnableExtend
\EnableExtendfalse

\newif\ifDoubleBlind 
\DoubleBlindfalse

\usepackage{algorithm}
\usepackage[noend]{algpseudocode}
\usepackage{graphicx}
\usepackage{balance}
\usepackage{mathtools}
\usepackage{numprint}
\usepackage[T1]{fontenc}
\usepackage{hyperref}
\usepackage{numprint}
\npthousandsep{.}

\let\oldReturn\Return
\renewcommand{\Return}{\State\oldReturn}
\ifVLDB

\newtheorem{lemma}{Lemma}

\vldbTitle{\papertitle}
\vldbAuthors{Monika Henzinger, Alexander Noe and Christian Schulz}
\vldbDOI{https://doi.org/10.14778/xxxxxxx.xxxxxxx}
\vldbVolume{47}
\vldbNumber{xxx}
\vldbYear{2021}

\else

\title{\papertitle}

\ifDoubleBlind
\author{Double Blind}{Double Blind University}{}{}{}
\authorrunning{Double Blind}
\Copyright{Double Blind}
\else\author{Monika Henzinger}{University of Vienna, Faculty of Computer Science, Vienna, Austria}{monika.henzinger@univie.ac.at}{0000-0002-5008-6530}{}
\author{Alexander Noe}{University of Vienna, Faculty of Computer Science, Vienna, Austria}{alexander.noe@univie.ac.at}{0000-0002-4711-3323}{}
\author{Christian Schulz}{Heidelberg University, Heidelberg, Germany}{christian.schulz@informatik.uni-heidelberg.de}{0000-0002-2823-3506}{}
\authorrunning{M. Henzinger, A. Noe and C. Schulz}
\Copyright{Monika Henzinger, Alexander Noe and Christian Schulz}
\fi{}

\nolinenumbers

\hideLIPIcs  

\EventEditors{John Q. Open and Joan R. Access}
\EventNoEds{2}
\EventLongTitle{42nd Conference on Very Important Topics (CVIT 2016)}
\EventShortTitle{CVIT 2016}
\EventAcronym{CVIT}
\EventYear{2016}
\EventDate{December 24--27, 2016}
\EventLocation{Little Whinging, United Kingdom}
\EventLogo{}
\SeriesVolume{42}
\ArticleNo{23}



\category{}

\relatedversion{}

\ifDoubleBlind

\else
\funding{
  The research leading to these results has received funding from the European Research Council under the
European Community's Seventh Framework Programme (FP7/2007-2013) /ERC grant agreement No. 340506. \newline Partially supported by DFG grant SCHU 2567/1-2.}
\fi{}

\fi{}
\begin{document}

{}

\maketitle

\begin{abstract}
We present a practically efficient algorithm for maintaining a global minimum cut in large dynamic graphs under both edge insertions and deletions. While there has been theoretical work on this problem, our algorithm is the first implementation of a fully-dynamic algorithm. The algorithm uses the theoretical foundation and combines it with efficient and finely-tuned implementations to give an algorithm that can maintain the global minimum cut of a graph with rapid update times. We show that our algorithm gives up to multiple orders of magnitude speedup compared to static approaches both on edge insertions and deletions.
\end{abstract}


\section{Introduction}

We consider the problem of maintaining a (global) \emph{minimum cut} of a graph
under \emph{edge insertions and deletions}, also known as the
\emph{fully-dynamic minimum cut problem}. A \emph{minimum cut} in a weighted graph in a
graph is a partition of the vertices into two sets so that the total weight of
edges connecting the sets is minimized. In the fully
dynamic setting, the algorithm has to process a sequence of edge insertions and
deletions and has to be able to return a minimum cut at any point in this sequence.

The minimum cut problem has applications in many fields, such as network reliability~\cite{karger2001randomized,ramanathan1987counting}, VLSI design~\cite{krishnamurthy1984improved}, graph drawing~\cite{kant1993algorithms}, as a subproblem in the branch-and-cut algorithm for solving the travelling salesperson problem and other combinatorial problems~\cite{padberg1991branch}, and as a subproblem in connectivity-based data reductions for problems such as cluster editing~\cite{Boecker2011}. Most real-world networks are continuously changing and evolving~\cite{demetrescu2010dynamic,eppstein1998dynamic,zaki2016comprehensive} and thus, dynamic algorithms that maintain a solution for a changing graph are of utmost importance for large-scale graph applications.

There has been a large body of research for the static minimum cut problem
starting in 1961~\cite{gomory1961multi}. The randomized algorithm of
Karger~\cite{karger2000minimum} with a running time of $\Oh{m \log^3{n}}$ is the
first algorithm with a quasi-linear running time. Kawarabayashi and
Thorup~\cite{kawarabayashi2015deterministic} give the first deterministic
quasi-linear algorithm, later improved by
Henzinger~\etal\cite{henzinger2017local} to a running time of $\Oh{m \log^2{n}
\log \log^2 n}$, which is the fastest deterministic minimum cut algorithm for
unweighted graphs.
Nagamochi~\etal\cite{nagamochi1992computing,nagamochi1994implementing} give an
algorithm for the minimum cut problem, which is based on edge contractions
instead of maximum flows. Their algorithm has a worst case running time of
$\Oh{nm+n^2\log{n}}$ but performs far better in practice on many graph
classes~\cite{Chekuri:1997:ESM:314161.314315,junger2000practical,henzinger2018practical}.
Gawrychowski~\etal\cite{gawrychowski2020minimum} give a randomized algorithm
that finds a minimum cut in an undirected weighted graph $G$ with high
probability in $\Oh{m\log^ 2 n}$ time, which is currently the fastest asymptotic
running time for the static minimum cut problem. Mukhopadhyay and
Nanongkai~\cite{DBLP:conf/stoc/MukhopadhyayN20} give a randomized algorithm with
running time $\Oh{m \frac{\log^2{n}}{\log \log n} + n \log^6 n}$. Li and
Panigrahi~\cite{li2020maxflows} give a deterministic algorithm that runs in time
$\Oh{m^{1+\epsilon}}$ plus polylog$(n)$ maximum flow computations for any
constant $\epsilon > 0$. Recently, Li~\cite{li2020deterministic} gave a
deterministic algorithm with running time
$\Oh{m^{1+\epsilon}}$ for any constant $\epsilon > 0$.

In the field of dynamic graph algorithms,
Henzinger~\cite{henzinger1995approximating} gives the first incremental minimum
cut algorithm, which maintains the exact minimum cut with an amortized time of
$\Oh{\lambda \log{n}}$ per edge insertion, where $\lambda$ is the value of the
minimum cut. Goranci~\etal\cite{goranci2018incremental} manage to remove the
dependence on $\lambda$ from the update time and give an incremental algorithm
with $\Oh{\log^3{n} \log \log^2{n}}$ amortized time per edge insertion. Both
algorithms maintain a compact data structure of all minimum cuts called
\emph{cactus graph} and invalidate minimum cuts whose weight was increased due
to an edge insertion. If there are no remaining minimum cuts, they recompute all
minimum cuts from scratch. For minimum cut values up to polylogarithmic size,
Thorup~\cite{thorup2007fully} gives a fully dynamic algorithm with
$\Otilde{\sqrt{n}}$ worst-case running time. Note that all of these algorithms
are limited to unweighted graphs. The algorithm of Thorup is based on greedy
tree packings using top trees.
Implementations~\cite{Chekuri:1997:ESM:314161.314315,DBLP:conf/swat/BhardwajLS20}
of static greedy tree packing algorithms give experimental results that are
significantly slower than implementations of minimum cut algorithms based on
edge contraction~\cite{nagamochi1992computing,nagamochi1994implementing} or
maximum flows~\cite{hao1992faster} on similar
graphs~\cite{Chekuri:1997:ESM:314161.314315,junger2000practical,henzinger2018practical}.

For \emph{planar} graphs with arbitrary
edge-weights, {\L}{\k{a}}cki and Sankowski~\cite{lkacki2011min} give a
fully-dynamic algorithm with $\Oh{n^{5/6}\log^{5/2}{n}}$ time per update. To the
best of our knowledge, there exists no implementation of any of these dynamic
algorithms. 

An important subproblem for many dynamic minimum cut algorithms is finding all minimum cuts. Even though a graph can have up to $\binom{n}{2}$ minimum cuts~\cite{karger2000minimum}, there is a compact representation of all minimum cuts of a graph called \emph{cactus graph} with at most $\Oh{n}$ vertices and edges. A cactus graph is a graph in which each edge belongs to at most one simple cycle. Karzanov and Timofeev~\cite{karzanov1986efficient} give the first polynomial time algorithm to construct the cactus representation for all minimum cuts. Picard and Queyranne~\cite{picard1980structure} show that all minimum s-t cuts, \ie minimum cuts separating two specified vertices can be found from a maximum flow between them. Thus, similar to the classical algorithm of Gomory and Hu~\cite{gomory1961multi} for the minimum cut problem, we can find all minimum cuts in $n-1$ maximum flow computations. The algorithm of Karzanov and Timofeev~\cite{karzanov1986efficient} combines all those minimum cuts into a cactus graph representing all minimum cuts. Nagamochi and Kameda~\cite{nagamochi1994canonical} give a representation of all minimum cuts separating two vertices $s$ and $t$ in a so-called $(s,t)$-cactus representation. Based on this $(s,t)$-cactus representation, Nagamochi~\etal\cite{nagamochi2000fast} give an algorithm that finds all minimum cuts and gives the minimum cut cactus in $\Oh{nm + n^2 \log{n} + n^*m\log{n}}$, where $n^*$ is the number of vertices in the cactus.

Karger and Stein~\cite{karger1996new} give a randomized algorithm to find all
minimum cuts in $\Oh{n^2 \log^3{n}}$ time by contracting random edges. Based on
the algorithm of Karzanov and Timofeev~\cite{karzanov1986efficient} and its
parallel variant given by Naor and Vazirani~\cite{naor1991representing} they
show how to give the cactus representation of the graph in the same asymptotic
time. Ghaffari~\etal\cite{ghaffari2019faster} gave an algorithm that
finds all \emph{non-trivial minimum cuts} of a simple unweighted graph in $\Oh{m
\log^2{n}}$ time. Using the techniques of Karger and Stein the algorithm can
trivially give the cactus representation of all minimum cuts in $\Oh{n^2
\log{n}}$. Recently, we developed an algorithm that
efficiently computes all minimum cuts of very large graphs in
practice~\cite{henzinger2020finding}. This algorithm combines various data
reductions with an efficient implementation of
the algorithm of Nagamochi~\etal\cite{nagamochi2000fast} and can find all
minimum cuts in graphs with up to billions of edges and millions of minimum cuts
in a few minutes.

\paragraph*{Our Results}

In this paper, we give the first implementation of a \emph{fully-dynamic
algorithm} for the \emph{minimum cut problem} in a weighted graph. Our algorithm maintains an exact
global minimum cut under edge insertions and deletions. For edge insertions, we
use the approach of
Henzinger~\cite{henzinger1995approximating} and
Goranci~\etal\cite{goranci2018incremental}, who maintain a compact data
structure of all minimum cuts in a graph and invalidate only the minimum cuts
that are affected by an edge insertion. We hereby use the recent algorithm of
Henzinger~\etal\cite{henzinger2020finding} to compute all minimum cuts in a
graph. For edge deletions, we use the push-relabel algorithm of Goldberg and
Tarjan~\cite{goldberg1988new} to certify whether the previous minimum cut is
still a minimum cut. As we only need to certify whether an edge deletion changes
the value of the minimum cut, we can perform optimizations that significantly
improve the speed of the push-relabel algorithm for our application. In
particular, we develop a fast initial labeling scheme and terminate early when
the connecitivity value is certified.

\section{Basic Concepts}

Let $G = (V, E, c)$ be a weighted undirected simple graph with vertex set $V$, edge set $E \subset V \times V$ and
non-negative edge weights $c: E \rightarrow \MdN$. 
We extend $c$ to a set of edges $E' \subseteq E$ by summing the weights of the edges; that is, let $c(E')\Is \sum_{e=(u,v)\in E'}c(u,v)$ and let $c(u)$ denote the sum of weights of all edges incident to vertex $v$.
Let $n = |V|$ be the
number of vertices and $m = |E|$ be the number of edges in $G$. The \emph{neighborhood}
$N(v)$ of a vertex $v$ is the set of vertices adjacent to $v$. The \emph{weighted degree} of a vertex is the sum of the weights of its incident edges. For brevity, we simply call this the \emph{degree} of the vertex.
For a set of vertices $A\subseteq V$, we denote by $E[A]\Is \{(u,v)\in E \mid u\in A, v\in V\setminus A\}$; that is, the set of edges in $E$ that start in $A$ and end in its complement.
A cut $(A, V
\setminus A)$ is a partitioning of the vertex set $V$ into two non-empty
\emph{partitions} $A$ and $V \setminus A$, each being called a \emph{side} of the cut. The \emph{capacity} or \emph{weight} of a cut $(A, V
\setminus A)$ is $c(A) = \sum_{(u,v) \in E[A]} c(u,v)$.
A \emph{minimum cut} is a cut $(A, V
\setminus A)$ that has smallest capacity $c(A)$ among all cuts in $G$. We use $\lambda(G)$ (or simply
$\lambda$, when its meaning is clear) to denote the value of the minimum cut
over all $A \subset V$. For two vertices $s$ and $t$, we denote $\lambda(G,s,t)$ as the capacity of the smallest cut of $G$, where $s$ and $t$ are on different sides of the cut. $\lambda(G,s,t)$ is also known as the \emph{minimum s-t-cut} of the graph. If all edges have weight $1$, $\lambda(G,s,t)$ is also called the \emph{connectivity} of vertices $s$ and $t$. The connectivity $\lambda(G,e)$ of an edge $e=(s,t)$ is defined as $\lambda(G,s,t)$, the connectivity of its incident vertices. At any point in the execution of a minimum cut algorithm,
$\hat\lambda(G)$ (or simply $\hat\lambda$) denotes the smallest upper bound of the
minimum cut that the algorithm discovered until that point. 
For a vertex $u \in V$ with minimum vertex degree, the size of the \emph{trivial cut} $(\{u\}, V\setminus \{u\})$ is equal to the vertex degree of $u$.
Many algorithms tackling the minimum cut problem use \emph{graph contraction}.
Given
an edge $e = (u, v) \in E$, we define $G/(u, v)$ (or $G/e$) to be the graph after \emph{contracting
edge} $(u, v)$. In the contracted graph, we delete vertex $v$ and all edges
incident to this vertex. For each edge $(v, w) \in E$, we add an edge $(u, w)$
with $c(u, w) = c(v, w)$ to~$G$ or, if the edge already exists, we give it the edge
weight $c(u,w) + c(v,w)$.

A graph with $n$ vertices can have up to $\Omega(n^2)$ minimum cuts~\cite{karger2000minimum}. To see that this bound is tight, consider an unweighted cycle with $n$ vertices. Each set of $2$ edges in this cycle is a minimum cut of $G$. This yields a total of $\binom{n}{2}$ minimum cuts. 
However, all minimum cuts can be represented by a cactus graph $\mathcal{C}$ with up to $2n$ vertices and $\Oh{n}$ edges~\cite{nagamochi2000fast}. A cactus graph is a connected graph, in which any two simple cycles have at most one vertex in common. In a cactus graph, each edge belongs to at most one simple cycle. 

To represent all minimum cuts of a graph $G$ in an edge-weighted cactus graph
$\mathcal{C} = (V(\mathcal{C}), E(\mathcal{C}))$, each vertex of $\mathcal{C}$
represents a possibly empty set of vertices of $G$ and each vertex in $G$
belongs to the set of one vertex in $\mathcal{C}$. Let $\Pi$ be a function that
assigns to each vertex of $\mathcal{C}$ is set of vertices of $G$. Then every
cut $(S, V(\mathcal{C}) \backslash S)$ in $\mathcal{C}$ corresponds to a minimum cut $(A, V
\backslash A)$ in $G$ where $A=\cup_{x\in S} \Pi(x)$. In $\mathcal{C}$, all
edges that do not belong to a cycle have weight $\lambda$ and all cycle edges
have weight $\frac{\lambda}{2}$. A minimum cut in $\mathcal{C}$ consists of
either one tree edge or two edges of the same cycle. We denote by $n^*$ the
number of vertices in $\mathcal{C}$ and $m^*$ the number of edges in
$\mathcal{C}$. The weight $c(v)$ of a vertex $v \in \mathcal{C}$ is equal to the
number of vertices in $G$ that are assigned to $v$.

\subsection{Push-relabel algorithm}

In this work we use and adapt the push-relabel algorithm of Goldberg and Tarjan~\cite{goldberg1988new} for the minimum $s$-$t$-cut problem. The algorithm aims to push as much flow as possible from the \emph{source} vertex $s$ to the \emph{sink} vertex $t$ and returns the value of the maximum flow between $s$ and $t$, which is equal to the value of the minimum cut separating them~\cite{dantzig2003max}. We now give a brief description of the algorithm, for more details we refer the reader to the original work~\cite{goldberg1988new}.

Let $G = (V,E,c)$ be a directed edge-weighted graph. An undirected edge $e=(u,v)$ is hereby interpreted as two symmetric directed edges $(u,v)$ and $(v,u)$ with $c(e) = c(u,v) = c(v,u)$. In the push-relabel algorithm, each vertex $v \in V$ has a \emph{distance} or \emph{height label} $d(x)$, initially $d(x) = 0$ for every vertex except $d(s) = n$. The algorithm handles a \emph{preflow}, a function $f$ so that for each edge $e$, $0 \geq f(e) \geq c(e)$ and for each $v \in V \backslash s$, $\sum_{(v,x) \in E} f((v,x)) \leq \sum_{(y,v)\in E} f((y,v))$ there is at least as much ingoing as outgoing flow. The difference in ingoing and outgoing flow in a vertex is called the \emph{excess flow} of this vertex.

First, the algorithm pushes flow from $s$ to all neighboring vertices, afterwards vertices push their excess flow to neighbors with a lower distance $d$. If a vertex $v$ has positive excess but no neighbors with a lower distance, the \emph{relabel} function increases the distance of $v$ until at least one outgoing preflow $f$ can be increased. At termination, the push-relabel algorithm reaches a \emph{flow}, where each edge $e$ has $0 \leq f(e) \leq c(e)$ units of flow and the excess of each vertex except $s$ and $t$ is $0$. The value of the minimum cut $\lambda(s,t)$ separating $s$ and $t$ is equal to the excess flow on $t$. Inherent to the push-relabel algorithm is the \emph{residual graph} $G_f = (V,E_f)$ for a given preflow $f$, where $E_f$ contains all edges $e = (u,v) \in E$ with $f(e) < c(e)$, \ie edges that have capacity to handle additional flow, and a reverse-edge for every edge where $0 < f(e)$.

\subsection{Finding All Minimum Cuts in a Graph}

Recently, we developed an algorithm to find all minimum cuts in a graph~\cite{henzinger2020finding} that can find all minimum cuts in graphs with up to billions of edges and millions of minimum cuts in mere minutes. The algorithm first uses a multitude of data-reduction techniques to find edges that can be contracted without affecting any minimum cut in the graph. On the reduced graph, the algorithm is based on the recursive algorithm of Nagamochi~\etal\cite{nagamochi2000fast}. In each step, the algorithm of Nagamochi~\etal\cite{nagamochi2000fast} selects a random edge $e = (u,v)$ and computes $\lambda(u,v)$, the smallest cut separating them. If $\lambda(u,v) = \lambda$, the edge is \emph{critical} and at least one minimum cut separates $u$ from $v$. The set of minimum $u$-$v$-cuts can be described as a set of vertex sets $(V_1,\dots,V_k)$ with $u \in V_1$ and $v \in V_k$, where for each $i \in [1,k-1]: (V_1\cup\dots\cup V_i, V_{i+1}\cup\dots\cup V_k)$ is a minimum cut. The algorithm then creates one instance for each vertex set $V_i$, where $V \backslash V_i$ is contracted into a single vertex and combines the cacti when leaving the recursion. The algorithm finds a cactus $\mathcal{C}$ representing all minimum cuts in time $\Oh{nm + n^2 \log{n} + n^*m\log{n}}$, but is significantly faster in practice on real-world instances.

\section{Fully Dynamic Minimum Cut}

In this section we develop an efficient fully-dynamic algorithm for the global
minimum cut. For this, we use techniques from a multitude of original works
combined with new and improved algorithmic solutions to engineer an algorithm
that is able to solve the dynamic minimum cut problem by orders of magnitude
faster than a static recomputation of the solution for a wide variety of graphs.
An important observation for dynamic minimum cut algorithms is that graphs often
have a large set of global minimum cuts~\cite{henzinger2020finding}. Thus,
dynamic minimum cut algorithms can avoid costly recomputation by storing a
compact data structure representing all minimum
cuts~\cite{henzinger1995approximating,goranci2018incremental} and only
invalidate changed cuts in edge insertion. The data structure we use is a
\emph{cactus graph}, \ie a graph in which every vertex is part of at most one
cycle. A minimum cut in the cactus graph is represented by either a tree edge or
two edges of the same cycle~\cite{henzinger2020finding}. For a graph with
multiple connected components, \ie a graph whose minimum cut value $\lambda =
0$, the cactus graph $\mathcal{C}$ has an empty edge set and one vertex
corresponding to each connected component.

The rest of this section is organized as follows: we start by explaining the
incremental minimum cut algorithm, followed by a description of the decremental minimum
cut algorithm and conclude by showing how to combine the routines to a fully
dynamic minimum cut algorithm.

\subsection{Incremental Minimum Cut}
\label{ss:insert}

For incremental minimum cut, our algorithm is closely related to the exact incremental dynamic algorithms of Henzinger~\cite{henzinger1995approximating} and Goranci~\etal\cite{goranci2018incremental}. On initialization of the algorithm with graph $G$, we run the recent algorithm of Henzinger~\etal\cite{henzinger2020finding} on $G$ to find the weight of the minimum cut $\lambda$ and the cactus graph $\mathcal{C}$ representing all minimum cuts in $G$. Each minimum cut in $\mathcal{C}$ corresponds to a minimum cut in $G$ and each minimum cut in $G$ corresponds to one or more minimum cuts in $\mathcal{C}$~\cite{henzinger1995approximating}.

The insertion of an edge $e=(u,v)$ with positive weight $c(e) > 0$ increases the weight of all cuts in which $u$ and $v$ are in different partitions, \ie in different vertices of the cactus graph $\mathcal{C}$. The weight of cuts in which $u$ and $v$ are in the same partition remains unchanged. As edge weights are non-negative, no cut weight can be decreased by inserting additional edges.

If $\Pi(u) = \Pi(v)$, \ie both vertices are mapped to the same vertex in $\mathcal{C}$, there is no minimum cut that separates $u$ and $v$ and all minimum cuts remain intact. If $\Pi(u) \neq \Pi(v)$, \ie the vertices are mapped to different vertices in $\mathcal{C}$, we need to invalidate the affected minimum cuts by contracting the corresponding edges in $\mathcal{C}$.

\paragraph*{Path Contraction}

Dinitz~\cite{dinitz1993maintaining} shows that for a connected graph with
$\lambda > 0$ the minimum cuts that are affected by the insertion of $(u,v)$
correspond to the minimum cuts on the path between $\Pi(u)$ and $\Pi(v)$. We
find the path using alternating breadth-first searches from $\Pi(u)$ and
$\Pi(v)$. For this path-finding algorithm, imagine the cactus graph
$\mathcal{C}$ as a tree graph in which each cycle is contracted into a single
vertex. On this tree, there is a unique path from $\Pi(u)$ to $\Pi(v)$.  

For every cycle in $\mathcal{C}$ that contains at least two vertices of the path
between $\Pi(u)$ and $\Pi(v)$, the cycle is ``squeezed'' by contracting the
first and last path vertex in the cycle, thus creating up to two new cycles.
Figure~\ref{fig:contract} shows an example in which a cycle is squeezed. In
Figure~\ref{fig:contract}, the cycle is squeezed by contracting the bottom left
and top right vertices. This creates a new cycle of size $3$ and a ``cycle'' of
size $2$, which is simply a new tree edge in the cactus graph $\mathcal{C}$. For
details and correctness proofs we refer the reader to the work of
Dinitz~\cite{dinitz1993maintaining}. The intuition is that due to the insertion
of the new edge, all cactus vertices in the path from $\Pi(u)$ and $\Pi(v)$ are now
connected with a value $> \lambda$, as their previous connection was $\lambda$
and the newly introduced edge increased it. For any cycle in the path, this also
includes the first and last cycle vertices $x$ and $y$ in the path, as these two vertices now
have a higher connectivity $\lambda(x,y)$. The minimum cuts that are represented
by edges in this cycle that have $x$ and $y$ on the same side are unaffected, as
all vertices in the path from $\Pi(u)$ and $\Pi(v)$ are on the same side of this
cut. As this is not true for cuts that separate $x$ and $y$, we merge $x$ and
$y$ (as well as the rest of the path from $\Pi(u)$ to $\Pi(v)$), which
``squeezes'' the cycle and creates up to two new cycles.

\begin{figure}[t!]
       \includegraphics[width=\textwidth]{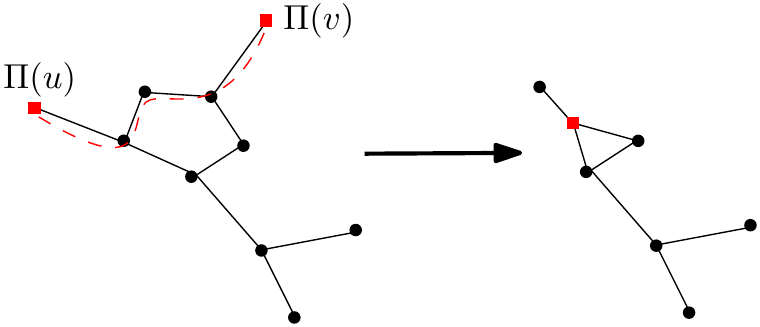}
       \caption{\label{fig:contract} Insertion of edge $e=(u,v)$ - contraction of path in $\mathcal{C}$, squeezing of cycle}
\end{figure}

If the graph has multiple connected components, \ie the graph has a minimum cut value $\lambda = 0$, $\mathcal{C}$ is a graph with no edges where each connected component is mapped to a vertex. The insertion of an edge between different connected components $\Pi(u)$ and $\Pi(v)$ merges the two vertices representing the connected components, as they are now connected.

If $\mathcal{C}$ has at least two non-empty vertices after the edge insertion, there is at least one minimum cut of value $\lambda$ remaining in the graph, as all minimum cuts that were affected by the insertion of edge $e$ were just removed from the cactus graph $\mathcal{C}$. As an edge insertion cannot decrease any connectivities, $\lambda$ remains the value of the minimum cut. If $\mathcal{C}$ only has a single non-empty vertex, we need to recompute the cactus graph $\mathcal{C}$ using the algorithm of Henzinger~\etal\cite{henzinger2020finding}.

Checking the set affiliation $\Pi$ of $u$ and $v$ can be done in constant time.
If $\Pi(u) = \Pi(v)$ and the cactus graph does not need to be updated, no
additional work needs to be done. If $\Pi(u) \neq \Pi(v)$, we perform
breadth-first search on $\mathcal{C}$ with $n^* \coloneqq |V(\mathcal{C})|$ and
$m^* \coloneqq |E(\mathcal{C})|$ which has a asymptotic running time of
$\Oh{n^* + m^*} = \Oh{n^*}$, contract the path from $\Pi(u)$ to $\Pi(v)$ in
$\Oh{n^*}$ and then update the set affiliation of all contracted vertices. This
update has a worst-case running time of $\Oh{n}$, however, contracting all
vertices of the path from $\Pi(u)$ to $\Pi(v)$ into the cactus graph vertex that
already corresponds to the most vertices of $G$, we often only need to update
the affiliation of a few vertices. Both the initial computation and a full
recomputation of the minimum cut cactus have a worst-case running time of
$\Oh{nm + n^2 \log{n} + n^*m\log{n}}$.

\subsection{Decremental Minimum Cut}
\label{ss:delete}

The deletion of an edge $e = (u,v)$ with positive weight $c(e) > 0$ decreases
the weight of all cuts in which $u$ and $v$ are in different partitions. This
might lead to a decrease of the minimum cut value $\lambda$ and thus the
invalidation of the minimum cuts in the existing minimum cut cactus
$\mathcal{C}$. The value of the minimum cut $\lambda(G,u,v)$ that separates
vertices $u$ and $v$ is equal to the maximum flow between them and can be found
by a variety of
algorithms~\cite{dinic1970algorithm,ford1956maximal,goldberg1988new}. In order
to check whether $\lambda$ is decreased by this edge deletion, we need to check
whether $\lambda(G-e,u,v) < \lambda(G)$. For this purpose, we use the
push-relabel algorithm of Goldberg and Tarjan~\cite{goldberg1988new} which aims
to push flow from $u$ to $v$ until there is no more possible path. 

\paragraph*{Early Termination}

We terminate the algorithm as soon as $\lambda(G)$ units of flow reached $v$. If $\lambda(G)$ units of flow from $u$ reached $v$, we know that $\lambda(G-e,u,v) \geq \lambda(G)$, \ie the connectivity of $u$ and $v$ on $G-e$ is at least as large as the minimum cut on $G$, the minimum cut value $\lambda$ remains unchanged. Note that iff $\lambda(G-e,u,v) = \lambda(G)$, the deletion of $e$ introduces one or more new minimum cuts. We do not introduce these new cuts to $\mathcal{C}$. The trade-off hereby is that we are able to terminate the push-relabel algorithm earlier and do not need to perform potentially expensive operations to update the cactus, but do not necessarily keep all cuts and have to recompute the cactus earlier. As most real-world graphs have a large number of minimum cuts~\cite{henzinger2020finding}, there are far more edge deletions than recomputations of $\mathcal{C}$.

Each edge deletion calls the push-relabel algorithm using the lowest-label selection
rule with a worst-case running time of $\Oh{n^2m}$~\cite{goldberg1988new}. The
lowest-label selection rule picks the active vertices whose distance label is
lowest, \ie a vertex that is close to the sink $v$. Using
highest-level selection would improve the worst-case running time to
$\Oh{n^2\sqrt{m}}$, but we aim to push as much flow as possible to the sink
early to be able to terminate the algorithm early as soon as $\lambda$ units of
flow reach the sink. Using lowest-level selection prioritizes the vertices close
to the sink and thus increases the amount of flow which reaches the sink at a
given point in time. Preliminary experiments show faster running times using the
lowest-level selection rule.

\subsubsection{Decremental Rebuild of Cactus Graph}

If the push-relabel algorithm finishes with a value of $< \lambda(G)$, we update the minimum cut value $\lambda(G-e)$ to $\lambda(G-e,u,v)$. As the minimum cut value changed by the deletion of $e$ and this deletion only affects cuts which contain $e$, we know that all minimum cuts of the updated graph $G-e$ separate $u$ and $v$. We use this information to significantly speed up the cactus construction. Instead of running the algorithm of Henzinger~\etal\cite{henzinger2020finding}, we run only the subroutine which is used to compute the $(u,v)$-cactus, \ie the cactus graph which contains all cuts that separate $u$ and $v$, as we know that all minimum cuts of $G-e$ separate $u$ and $v$. This routine, developed by Nagamochi and Kameda~\cite{nagamochi1996constructing}, finds a $u$-$v$-cactus a running time of $\Oh{n+m}$.

Note that the routine of Nagamochi and Kameda~\cite{nagamochi1996constructing}
only guarantees to find all minimum $u$-$v$-cuts if an edge $e = (u,v)$ with
$c(e) > 0$ exists (\cite[Lemma 3.4]{nagamochi1996constructing}). As this edge
was just deleted in $G-e$ and therefore does not exist, it is possible that
\emph{crossing} $u$-$v$-cuts $(X,\overline X)$ and $(Y,\overline Y)$ with $u \in
X$ and $u \in Y$ exist. Two cuts are \emph{crossing}, if both $(\overline X \cap
Y)$ and $(Y \cap \overline X)$ are not empty. As we only find one cut in a pair
of crossing cuts, the $u$-$v$-cactus is not necessarily maximal. However, the
operation is significantly faster than recomputing the complete minimum cut
cactus in which almost all edges are not part of any minimum cut. While it is
not guaranteed that the decremental rebuild algorithm finds all minimum cuts in
$G-e$, every cut of size $\lambda(G-e,u,v)$ that is found is a minimum cut.
As we build the minimum cut cactus out of minimum cuts, it is a valid (but
potentially incomplete) minimum cut cactus and the algorithm is correct.

\subsubsection{Local Relabeling}

Many efficient implementations of the push-relabel algorithm use the global relabeling heuristic~\cite{cherkassky1997implementing} in order to direct flow towards the sink more efficiently. The push-relabel algorithm maintains a distance label $d$ for each vertex to indicate the distance from that vertex to the sink using only edges that can receive additional flow. The global relabeling heuristic hereby periodically performs backward breadth-first search to compute distance labels on all vertices.

This heuristic can also be used to set the initial distance labels in the flow
network for a flow problem with source $u$ and sink $v$. This has a running time
of $\Oh{n + m}$ but helps lead the flow towards the sink. As our algorithm
terminates the push-relabel algorithm early, we try to avoid the $\Oh{m}$
running time while still giving the flow some guidance. Thus, we perform
\emph{local relabeling} with a \emph{relabeling depth} of $\gamma$ for $\gamma
\in [0, n)$, where we set $d(v) = 0$, $d(u) = n$ and then perform a backward
breadth-first search around the sink $v$, in which we set $d(x)$ to the length
of the shortest path between $x$ and $v$ (at this point, there is no flow in the
network, so every edge in $G$ is admissible). Instead of setting the distance of
every vertex, we only explore the neighborhoods of vertices $x$ with $d(x) <
\gamma$, thus we only set the distance-to-sink for vertices with $d(x) \leq
\gamma$. For every vertex $y$ with a higher distance, we set $d(y) = (\gamma +
1)$. This results in a running time for setting the distance labels of $\Oh{n}$ 
plus the time needed to perform the bounded-depth breadth-first search.

This process creates a ``funnel'' around the sink to lead flow towards it,
without incurring a running time overhead of $\Theta(m)$ (if $\gamma$ is set
sufficiently low). Note that this is useful because the push-relabel algorithm
is terminated early in many cases and thus initializing the distance labels
faster can give a large speedup. We give experimental
results for different relabeling depths $\gamma$ for local relabeling in our
application in Section~\ref{ss:exp-relabel}.

\paragraph*{Correctness}

Goldberg and Tarjan show that each push and relabel operation in the push-relabel algorithm preserve a \emph{valid labeling}~\cite{goldberg1988new}. A valid labeling is a labeling $d$, where in a given preflow $f$ and corresponding residual graph $G_f$, for each edge $e = (u,v) \in E_f$, $d(u) \leq d(v) + 1$. We therefore need to show that the labeling $d$ that is given by the initial local relabeling is a valid labeling.

\begin{lemma} \label{lem:valid-flow}
    Let $G=(V,E,c)$ be a flow-graph with source $s$ and sink $t$ and let $d$ be the vertex labeling given by the local relabeling algorithm. The vertex labeling $d$ is a valid labeling.
\end{lemma}

\begin{proof}
       The vertex labeling $d$ is generated using breadth-first search. Thus, for every edge $e = (u,v)$ where $u \neq s$ and $v \neq s$, $|d(u) - d(v)| \leq 1$. We prove this by contradiction. W.l.o.g. assume that $d(u) - d(v) > 1$. As $u \neq s$ and $s$ is the only vertex with $d(s) > \gamma$, $d(u) \leq \gamma + 1$ and $d(v) < \gamma$. Thus, at some point of the breadth-first search, we set the distance labels of all neighbors of $v$ that do not yet have a distance label to $d(v) + 1$. As edge $e = (u,v)$ exists, $u$ and $v$ are neighbors and the labeling sets $d(u) = d(v) + 1$. This contradicts $d(u) - d(v) > 1$. 

       This shows that the labeling is valid for every edge not incident to the source $s$, as distance labels of incident non-source vertices differ by at most $1$. The only edges we need to check are edges incident to $s$. In the initialization of the push-relabel algorithm, all outgoing edges of the source $s$ are fully saturated with flow and are thus no outgoing edge of $s$ is in $E_f$. For ingoing edges $e = (v,s)$, we know that $0 \leq d(v) \leq \gamma + 1 = n$ and thus know that $d(v) \leq d(s)$. Thus $e$ respects the validity of labeling $d$.
\end{proof}

Lemma~\ref{lem:valid-flow} shows that local relabeling gives a valid labeling; which is upheld by the operations in the push-relabel algorithm~\cite{goldberg1988new}. Thus, correctness of the modified algorithm follows from the correctness proof of Goldberg and Tarjan.

Resetting the vertex data structures can be performed in $\Oh{n}$, however there are $m$ edges whose current flow needs to be reset to $0$. Using early termination we hope to solve some problems very fast in practice, as we can sometimes terminate early without exploring large parts of the graph. Thus, resetting of the edge flows in $\Oh{m}$ is a significant problem and is avoided using implicit resetting as described in the following paragraph.

Each flow problem that is solved over the course of the dynamic minimum cut
algorithm is given a unique ID, starting at an arbitrary integer and incrementing
from there.  In addition to the current flow on an edge, we also store the ID of
the last problem which accessed the flow on this edge. When the flow of an edge
is read or updated in a flow problem, we check whether the ID of the last access
equals the ID of the current problem. If they are equal, we simply return or
update the flow value, as the edge has already been accessed in this flow
problem and does not need to be reset. Otherwise, we need to reset the edge flow
to $0$ and set the problem ID to the ID of the current problem and then perform
the operation on the updated edge. Thus, we implicitly reset the edge flow on
first access in the current problem. As we increment the flow problem ID after
every flow problem, no two flow problems share the same ID.

Using this implicit reset of the edge flows saves $\Oh{m}$ overhead but
introduces a constant amount of work on each access and
update of the edge flow. It is therefore useful in practice if the problem
terminates with significantly fewer than $m$ flow updates due to early
termination. It does not affect the worst-case running time of the algorithm, as
we only perform a constant amount of work on each edge update. The running time
of the initialization of the implementation is improved from $\Oh{n+m}$ to
$\Oh{n}$, as we do not explicitly reset the flow on each edge.

\subsection{Fully Dynamic Minimum Cut}
\label{ss:fully}

Based on the incremental and decremental algorithm, we now describe our fully
dynamic algorithm. As the operations in the previous section each output the
minimum cut $\lambda(G)$ and a corresponding cut cactus $\mathcal{C}$ that
stores a set of minimum cuts for $G$, the algorithm gives correct results on all
operations. However, there are update sequences in which every insertion or
deletion changes the minimum cut value and, thus, triggers a recomputation
of the minimum cut cactus $\mathcal{C}$. One such example is the repeated
deletion and reinsertion of an edge that belongs to a minimum cut. In the
following paragraphs we describe a technique that is used to mitigate such
worst-case instances. Nevertheless, it is still possible to construct update
sequences in which the minimum cut cactus $\mathcal{C}$ needs to be recomputed
every $\Oh{1}$ edge updates and thus the worst-case asymptotic running
time per update is equal to the running time of the static algorithm. 

\subsubsection{Cactus Cache}
\label{ss:cache}

Computing the minimum cut cactus $\mathcal{C}$ is expensive if there is a large set of minimum cuts and the cactus is therefore large. Thus, it is beneficial to reduce the amount of recomputations to speed up the process. On some fully dynamic workloads, the minimum cut often jumps between values $\lambda_1$ and $\lambda_2$ with $\lambda_1 > \lambda_2$, where the minimum cut cactus for cut value $\lambda_1$ is large and thus expensive to recompute whenever the cut value changes.

A simple example workload is a large unweighted cycle, which has a minimum cut
of $2$. If we delete any edge, the minimum cut value changes to $1$, as the
incident vertices have a degree of $1$. By reinserting the just-deleted edge,
the minimum cut changes to a value of $2$ again and the minimum cut cactus is
equal to the cactus prior to the edge deletion. Thus we can save a
significant amount of work by caching and reusing the previous cactus graph when
the minimum cut is increased to $2$ again.

\paragraph*{Reuse Cactus Graph from Cache}

Whenever the deletion of an edge $e$ from graph $G$ decreases the minimum cut value from $\lambda_1$ to $\lambda_2$, we cache the previous cactus $\mathcal{C}$. After this point, we also remember all edge insertions, as these can invalidate minimum cuts in $\mathcal{C}$. If at a later point the minimum cut is again increased from $\lambda_2$ to $\lambda_1$ and the number of edge insertions divided by the number of vertices in $\mathcal{C}$ is smaller than a parameter $\delta$, we recreate the cactus graph from the cache instead of recomputing it. The default value for $\delta$ is $2$. The algorithm does not store the intermediate edge deletion, as there can only lower connectivities and by computing the minimum cut value we know that there is no cut of value $< \lambda_1$ and thus all cuts of value $\lambda_1$ are global minimum cuts.

For each edge insertion since caching we perform the edge insertion operation from Section~\ref{ss:insert} to eliminate all cuts that are invalidated by the edge insertion. All cuts that remain in $\mathcal{C}$ are still minimum cuts. If there are only a small amount of edge insertions since the cactus was cached, this is significantly faster than recomputing the cactus from scratch. As we do not remember edge deletions, the cactus might not contain all minimum cuts and thus require slightly earlier recomputation.

\section{Experiments and Results} \label{s:experiments}

We now perform an experimental evaluation of the proposed algorithms.
This is done in the following order. In Section~\ref{ss:instances} we show and detail the
static and dynamic graph instances that we use in our experimental evaluation.
In Section~\ref{ss:exp-relabel} we analyze the impact of local relabeling on the
static preflow-push algorithm to determine with value of the relabeling depth to
use in the experiments on dynamic graphs. Then
(Sections~\ref{ss:exp-dynamic}~and~\ref{ss:exp-static}) we evaluate our dynamic
algorithms on a wide variety of instances. In Section~\ref{ss:exp-worstcase} we
then generate a set of worst-case problems and use these to evaluate the
performance of our algorithm on instances that were created in order to be
difficult for them.

\paragraph*{Experimental Setup and Methodology}

We implemented the algorithms using \CC-17 and compiled all code using g++ version 8.3.0 with full optimization (\texttt{-O3}). Our experiments are conducted on a machine with two Intel Xeon Gold 6130 processors with 2.1GHz with 16 CPU cores each and $256$ GB RAM in total. In this section we first describe our experimental methodology. Afterwards, we evaluate different algorithmic choices in our algorithm and then we compare our algorithm to the state of the art. When we report a mean result we give the geometric mean as problems differ significantly in cut size and time. Our code is freely available under the permissive MIT license\ifDoubleBlind{\footnote{Link removed, as this submission is double-blind}}\else{\footnote{\url{https://github.com/VieCut/VieCut}}}\fi{}. 

\subsection{Graph Instances}
\label{ss:instances}

In our experiments, we use a wide variety of static and dynamic graph instances. These are social graphs, web graphs, co-purchase matrices, cooperation networks and some generated instances.  All instances are undirected. If the original graph is directed, we generate an undirected graph by removing edge directions and then removing duplicate edges. In our experiments, we use three families of graphs.

\begin{table*}[!ht]
       \setlength\intextsep{0pt}
       \centering
       \resizebox*{.45\textwidth}{!}{
       \begin{tabular}{l||r|r|r|r|r}
              \hline
       \multicolumn{6}{c}{Graph Family A}\\\hline
         Graph & $n$ & $m$ & $\lambda$ & $\delta$ & $n^*$ \\\hline\hline
         \texttt{com-orkut} & $2.4M$ & $112M$ & \numprint{14} & \numprint{16} & \numprint{2} \\
         & \numprint{114190} & $18M$ & \numprint{89} & \numprint{95} & \numprint{2} \\
         & \numprint{107486} & $17M$ & \numprint{76} & \numprint{98} & \numprint{2} \\
         & \numprint{103911} & $17M$ & \numprint{70} & \numprint{100} & \numprint{2} \\\hline
         \texttt{eu-2005} & \numprint{605264} & $15M$ & \numprint{1} & \numprint{10} & \numprint{63} \\
         & \numprint{271497} & $10M$ & \numprint{2} & \numprint{25} & \numprint{3} \\
         & \numprint{58829} & $3.7M$ & \numprint{29} & \numprint{60} & \numprint{2} \\
         & \numprint{5289} & \numprint{464821} & \numprint{19} & \numprint{100} & \numprint{2}\\\hline
         \texttt{gsh-2015-host}  & $25M$ & $1.3B$ & \numprint{1} & \numprint{10} & \numprint{175}\\
          & $5.3M$ & $944M$ & \numprint{1} & \numprint{50} & \numprint{32}\\
               & $2.6M$ & $778M$ & \numprint{1} & \numprint{100} & \numprint{16}\\
              & \numprint{98275} & $188M$ & \numprint{1} & \numprint{1000} & \numprint{3} \\\hline
       \texttt{hollywood-2011} & $1.3M$ & $109M$ & \numprint{1} & \numprint{20} & \numprint{13}\\
        & \numprint{576111} & $87M$ & \numprint{6} & \numprint{60} & \numprint{2} \\
               & \numprint{328631} & $71M$ & \numprint{77} & \numprint{100} & \numprint{2}\\
               & \numprint{138536} & $47M$ & \numprint{27} & \numprint{200} & \numprint{2}\\\hline
               \texttt{twitter-2010}  & $13M$ & $958M$ & \numprint{1} & \numprint{25} & \numprint{2} \\
                 & $10M$ & $884M$ & \numprint{1} & \numprint{30} & \numprint{3} \\
                    & $4.3M$ & $672M$ & \numprint{3} & \numprint{50} & \numprint{3}\\
                     & $3.5M$ & $625M$ & \numprint{3} & \numprint{60} & \numprint{2}\\\hline
         \texttt{uk-2002} & $9M$ & $226M$ & \numprint{1} & \numprint{10} & \numprint{1940} \\
          & $2.5M$ & $115M$ & \numprint{1} & \numprint{30} & \numprint{347}\\
               & \numprint{783316} & $51M$ & \numprint{1} & \numprint{50} & \numprint{138}\\
               & \numprint{98275} & $11M$ & \numprint{1} & \numprint{100} & \numprint{20} \\\hline
         \texttt{uk-2007-05} & $68M$ & $3.1B$ & \numprint{1} & \numprint{10} & \numprint{3202}\\
         & $16M$ & $1.7B$ & \numprint{1} & \numprint{50} & \numprint{387} \\
               & $3.9M$ & $862M$ & \numprint{1} & \numprint{100} & \numprint{134}\\
               & \numprint{223416} & $183M$ & \numprint{1} & \numprint{1000} & \numprint{2}\\\hline

              \hline\hline
              \multicolumn{6}{c}{Graph Family B}\\\hline
    \texttt{amazon} & \numprint{64813} & \numprint{153973} & \numprint{1}& \numprint{1} & \numprint{10068} \\\hline
    \texttt{auto} & \numprint{448695} & $3.31M$ & \numprint{4} & \numprint{4} & \numprint{43} \\
    & \numprint{448529} & $3.31M$ & \numprint{5} & \numprint{5} & \numprint{102} \\
    & \numprint{448037} & $3.31M$ & \numprint{6} & \numprint{6} & \numprint{557} \\
    & \numprint{444947} & $3.29M$ & \numprint{7} & \numprint{7} & \numprint{1128} \\
    & \numprint{437975} & $3.24M$ & \numprint{8} & \numprint{8} & \numprint{2792} \\
    & \numprint{418547} & $3.10M$ & \numprint{9} & \numprint{9} & \numprint{5814} \\ \hline
    \texttt{caidaRouterLevel} & \numprint{190914} & \numprint{607610} & \numprint{1} & \numprint{1} & \numprint{49940} \\\hline
    \texttt{cfd2} & \numprint{123440} & $1.48M$ & \numprint{7} & \numprint{7} & \numprint{15} \\\hline
    \texttt{citationCiteseer} & \numprint{268495} & $1.16M$ & \numprint{1} & \numprint{1} & \numprint{43031} \\
     & \numprint{223587} & $1.11M$ & \numprint{2} & \numprint{2} & \numprint{33423} \\
     & \numprint{162464} & \numprint{862237} & \numprint{3} & \numprint{3} & \numprint{23373} \\
     & \numprint{109522} & \numprint{435571} & \numprint{4} & \numprint{4} & \numprint{16670} \\
     & \numprint{73595} & \numprint{225089} & \numprint{5} & \numprint{5} & \numprint{11878} \\
     & \numprint{50145} & \numprint{125580} & \numprint{6} & \numprint{6} & \numprint{8770} \\\hline
     \texttt{cnr-2000} & \numprint{325557} & $2.74M$ & \numprint{1} & \numprint{1} & \numprint{87720} \\
    & \numprint{192573} & $2.25M$ & \numprint{2} & \numprint{2} & \numprint{33745} \\
    & \numprint{130710} & $1.94M$ & \numprint{3} & \numprint{3} & \numprint{11604} \\
    & \numprint{110109} & $1.83M$ & \numprint{4} & \numprint{4} & \numprint{9256} \\
    & \numprint{94664} & $1.77M$ & \numprint{5} & \numprint{5} & \numprint{4262} \\
    & \numprint{87113} & $1.70M$ & \numprint{6} & \numprint{6} & \numprint{5796} \\ 
    & \numprint{78142} & $1.62M$ & \numprint{7} & \numprint{7} & \numprint{3213} \\
    & \numprint{73070} & $1.57M$ & \numprint{8} & \numprint{8} & \numprint{2449} \\\hline
    \texttt{coAuthorsDBLP} & \numprint{299067} & \numprint{977676} & \numprint{1} & \numprint{1} & \numprint{45242} \\\hline
    \texttt{cs4} & \numprint{22499} & \numprint{43858} & \numprint{2} & \numprint{2} & \numprint{2} \\\hline
    \texttt{delaunay\_n17} & \numprint{131072} & \numprint{393176} & \numprint{3} & \numprint{3} & \numprint{1484} \\ \hline
    \texttt{fe\_ocean} & \numprint{143437} & \numprint{409593} & \numprint{1} & \numprint{1} & \numprint{40} \\\hline
    \texttt{kron-logn16} & \numprint{55319} & $2.46M$ & \numprint{1} & \numprint{1} & \numprint{6325} \\\hline
    \texttt{luxembourg} & \numprint{114599} & \numprint{239332} & \numprint{1} & \numprint{1} & \numprint{23077} \\\hline
    \texttt{vibrobox} & \numprint{12328} & \numprint{165250} & \numprint{8} & \numprint{8} & \numprint{625} \\\hline
    \texttt{wikipedia} & \numprint{35579} & \numprint{495357} & \numprint{1} & \numprint{1} & \numprint{2172} \\\hline
       \end{tabular}
       }
       \quad
       \resizebox*{.48\textwidth}{!}{
       \begin{tabular}{l||r|r|r|r|r}
       \hline
         \multicolumn{6}{c}{Graph Family B (continued)}\\\hline
         Graph & $n$ & $m$ & $\lambda$ & $\delta$ & $n^*$ \\\
         \texttt{amazon-2008} & \numprint{735323} & $3.52M$ & \numprint{1} & \numprint{1} & \numprint{82520}\\
         & \numprint{649187} & $3.42M$ & \numprint{2} & \numprint{2} & \numprint{50611}\\
         & \numprint{551882} & $3.18M$ & \numprint{3} & \numprint{3} & \numprint{35752}\\
         & \numprint{373622} & $2.12M$ & \numprint{5} & \numprint{5} & \numprint{19813}\\
         & \numprint{145625} & \numprint{582314} & \numprint{10} & \numprint{10} & \numprint{64657}\\\hline
         \texttt{coPapersCiteseer} & \numprint{434102} & $16.0M$ & \numprint{1} & \numprint{1} & \numprint{6372} \\
         & \numprint{424213} & $16.0M$ & \numprint{2} & \numprint{2} & \numprint{7529} \\
         & \numprint{409647} & $15.9M$ & \numprint{3} & \numprint{3} & \numprint{7495} \\
         & \numprint{379723} & $15.5M$ & \numprint{5} & \numprint{5} & \numprint{6515} \\
         & \numprint{310496} & $13.9M$ & \numprint{10} & \numprint{10} & \numprint{4579} \\\hline
         \texttt{eu-2005} & \numprint{862664} & $16.1M$ & \numprint{1} & \numprint{1} & \numprint{52232} \\
         & \numprint{806896} & $16.1M$ & \numprint{2} & \numprint{2} & \numprint{42151} \\
         & \numprint{738453} & $15.7M$ & \numprint{3} & \numprint{3} & \numprint{21265} \\
         & \numprint{671434} & $13.9M$ & \numprint{5} & \numprint{5} & \numprint{18722} \\
         & \numprint{552566} & $11.0M$ & \numprint{10} & \numprint{10} & \numprint{23798} \\\hline
         \texttt{hollywood-2009} & $1.07M$ & $56.3M$ & \numprint{1} & \numprint{1} & \numprint{11923} \\
         & $1.06M$ & $56.2M$ & \numprint{2} & \numprint{2} & \numprint{17386} \\
         & $1.03M$ & $55.9M$ & \numprint{3} & \numprint{3} & \numprint{21890} \\
         & \numprint{942687} & $49.2M$ & \numprint{5} & \numprint{5} & \numprint{22199} \\
         & \numprint{700630} & $16.8M$ & \numprint{10} & \numprint{10} & \numprint{19265} \\\hline
         \texttt{in-2004} & $1.35M$ & $13.1M$ & \numprint{1} & \numprint{1} & \numprint{278092} \\
         & \numprint{909203} & $11.7M$ & \numprint{2} & \numprint{2} & \numprint{89895} \\
         & \numprint{720446} & $9.2M$ & \numprint{3} & \numprint{3} & \numprint{45289} \\
         & \numprint{564109} & $7.7M$ & \numprint{5} & \numprint{5} & \numprint{33428} \\
         & \numprint{289715} & $5.1M$ & \numprint{10} & \numprint{10} & \numprint{12947} \\\hline
         \texttt{uk-2002} & $18.4M$ & $261.6M$ & \numprint{1} & \numprint{1} & $2.5M$ \\
         & $15.4M$ & $254.0M$ & \numprint{2} & \numprint{2} & $1.4M$ \\
         & $13.1M$ & $236.3M$ & \numprint{3} & \numprint{3} & \numprint{938319} \\
         & $10.6M$ & $207.6M$ & \numprint{5} & \numprint{5} & \numprint{431140} \\
         & $7.6M$ & $162.1M$ & \numprint{10} & \numprint{10} & \numprint{298716} \\
         & \numprint{657247} & $26.2M$ & \numprint{50} & \numprint{50} & \numprint{24139} \\
         & \numprint{124816} & $8.2M$ & \numprint{100} & \numprint{100} & \numprint{3863} \\
         \hline\hline
         \multicolumn{6}{c}{Graph Family C}\\\hline
         Dynamic Graph & $n$ & Insertions & Deletions & Batches & $\lambda$ \\ \hline\hline
         \texttt{aves-weaver-social} & \numprint{445} & \numprint{1423} & \numprint{0} & \numprint{23} & \numprint{0} \\
         \texttt{ca-cit-HepPh} & \numprint{28093} & $4.60M$ & \numprint{0} & \numprint{2337} & \numprint{0} \\
         \texttt{ca-cit-HepTh} & \numprint{22908} & $2.67M$ & \numprint{0} & \numprint{219} & \numprint{0} \\
         \texttt{comm-linux-kernel-r} & \numprint{63399} & $1.03M$ & \numprint{0} & \numprint{839643} & \numprint{0} \\
         \texttt{copresence-InVS13} & \numprint{987} & \numprint{394247} & \numprint{0}& \numprint{20129} & \numprint{0} \\
         \texttt{copresence-InVS15} & \numprint{1870} & $1.28M$ & \numprint{0} & \numprint{21536} & \numprint{0} \\
         \texttt{copresence-LyonS} & \numprint{1922} & $6.59M$ & \numprint{0} & \numprint{3124} & \numprint{0} \\
         \texttt{copresence-SFHH} & \numprint{1924} & $1.42M$ & \numprint{0} & \numprint{3149} & \numprint{0} \\
         \texttt{copresence-Thiers} & \numprint{1894} & $18.6M$ & \numprint{0} & \numprint{8938} & \numprint{0} \\
         \texttt{digg-friends} & \numprint{279630} & $1.73M$ & \numprint{0} & $1.64M$ & \numprint{0} \\
         \texttt{edit-enwikibooks} & \numprint{134942} & $1.16M$ & \numprint{0} & $1.13M$ & \numprint{0} \\
         \texttt{fb-wosn-friends} & \numprint{63731} & $1.27M$ & \numprint{0} & \numprint{736675} & \numprint{0} \\
         \texttt{ia-contacts\_dublin} & \numprint{10972} & \numprint{415912} & \numprint{0} & \numprint{76944} & \numprint{0} \\
         \texttt{ia-enron-email-all} & \numprint{87273} & $1.13M$ & \numprint{0} & \numprint{214908} & \numprint{0} \\
         \texttt{ia-facebook-wall} & \numprint{46952} & \numprint{855542} & \numprint{0} & \numprint{847020} & \numprint{0} \\
         \texttt{ia-online-ads-c} & $15.3M$ & \numprint{133904} & \numprint{0} & \numprint{56565} & \numprint{0} \\
         \texttt{ia-prosper-loans} & \numprint{89269} & $3.39M$ & \numprint{0} & \numprint{1259} & \numprint{0} \\
         \texttt{ia-stackexch-user} & \numprint{545196} & $1.30M$ & \numprint{0} & \numprint{1154} & \numprint{1} \\
         \texttt{ia-sx-askubuntu-a2q} & \numprint{515273} & \numprint{257305} & \numprint{0} & \numprint{257096} & \numprint{0} \\
         \texttt{ia-sx-mathoverflow} & \numprint{88580} & \numprint{390441} & \numprint{0} & \numprint{390051} & \numprint{0} \\
         \texttt{ia-sx-superuser} & \numprint{567315} & $1.11M$ & \numprint{0} & $1.10M$ & \numprint{0} \\
         \texttt{ia-workplace-cts} & \numprint{987} & \numprint{9827} & \numprint{0} & \numprint{7104} & \numprint{0} \\
         \texttt{imdb} & \numprint{150545} & \numprint{296188} & \numprint{0} & \numprint{7104} & \numprint{0} \\
         \texttt{insecta-ant-colony1} & \numprint{113} & \numprint{111578} & \numprint{0} & \numprint{41} & \numprint{4285} \\
         \texttt{insecta-ant-colony2} & \numprint{131} & \numprint{139925} & \numprint{0} & \numprint{41} & \numprint{3742} \\
         \texttt{insecta-ant-colony3} & \numprint{160} & \numprint{241280} & \numprint{0} & \numprint{41} & \numprint{1539} \\
         \texttt{insecta-ant-colony4} & \numprint{102} & \numprint{81599} & \numprint{0} & \numprint{41} & \numprint{1838} \\
         \texttt{insecta-ant-colony5} & \numprint{152} & \numprint{194317} & \numprint{0} & \numprint{41} & \numprint{6671} \\
         \texttt{insecta-ant-colony6} & \numprint{164} & \numprint{247214} & \numprint{0} & \numprint{39} & \numprint{2177} \\
         \texttt{mammalia-voles-kcs} & \numprint{1218} & \numprint{4258} & \numprint{0} & \numprint{64} & \numprint{0} \\
         \texttt{SFHH-conf-sensor} & \numprint{1924} & \numprint{70261} & \numprint{0} & \numprint{3509} & \numprint{0} \\
         \texttt{soc-epinions-trust} & \numprint{131828} & \numprint{717129} & \numprint{123670} & \numprint{939} & \numprint{0} \\
         \texttt{soc-flickr-growth} & $2.30M$ & $33.1M$ & \numprint{0} & \numprint{134}& \numprint{0} \\
         \texttt{soc-wiki-elec} & \numprint{8297} & \numprint{83920} & \numprint{23093} & \numprint{101014} & \numprint{0} \\
         \texttt{soc-youtube-growth} & $3.22M$ & $12.2M$ & \numprint{0} & \numprint{203} & \numprint{0} \\
         \texttt{sx-stackoverflow} & $2.58M$ & \numprint{392515} & \numprint{0} & \numprint{384680} & \numprint{0} \\
         \hline
       \end{tabular}
       }
       \caption{Statistics of static and dynamic graphs used in experiments.\label{table:graphs}}
     \end{table*}

Family A consists of $28$ graphs obtained from the work of Henzinger~\etal on
the global minimum cut problem~\cite{henzinger2018practical}, originally from
the 10th DIMACS Implementation challenge~\cite{bader2013graph} and the
SuiteSparse Matrix Collection~\cite{davis2011university}. As finding \emph{any}
minimum cut is easy if the minimum cut is equal to the minimum degree, these
graphs are \emph{k-cores} of large social networks in which the minimum cut is
strictly smaller than the minimum degree. A \emph{k-core} of a graph is a
subgraph in which we iteratively remove all vertices of degree $<k$ until the
graph has a minimum degree of $k$. If the k-core has multiple connected
components, we use the largest of them. In this graph family, there are
generally only one or a few minimum cuts on each graph and the minimum cut is
strictly smaller than the minimum degree. In Table~\ref{table:graphs} we report
both the minimum cut $\lambda$ and the minimum degree $\delta$. In the dataset
there are $7$ different graphs with each $4$ different values of $\delta = k$
with $\lambda < \delta$ in every instance.

Family B consists of $65$ graphs obtained from the work of Henzinger~\etal on
finding all minimum cuts of a graph~\cite{henzinger2020finding}, originally from
from the 10th DIMACS Implementation challenge~\cite{bader2013graph}, the
SuiteSparse Matrix Collection~\cite{davis2011university} and the Walshaw Graph
Partitioning Archive~\cite{soper2004combined}. They represent a wide variety of
real-world graphs from different fields and applications. In contrast to Family
A, most graphs in this graph family have a large number of minimum cuts and
generally the minimum cut is equal to the minimum degree. As most real-world
graphs have some vertices of very low degree and therefore also a low minimum
cut, we also create instances with higher minimum cut by computing beforehand
the largest subgraph $G_x$ of $G$ that does not contain any minimum cut of size
$\lambda(G)$.

Family C consists of a set of $36$ dynamic graphs from Network Repository~\cite{nr-aaai15,nr-sigkdd16}. These graphs consist of a sequence of edge insertions and deletions. While edges are inserted and deleted, all vertices are static and remain in the graph for the whole time. Each edge update has an associated timestamp, a set of updates with the same timestamp is called a \emph{batch}. Most of the graphs in this dataset have multiple connected components, \ie their minimum cut $\lambda$ is $0$.

\subsection{Local Relabeling}
\label{ss:exp-relabel}

\begin{figure}[t!]
       \centering
       \includegraphics[width=\textwidth]{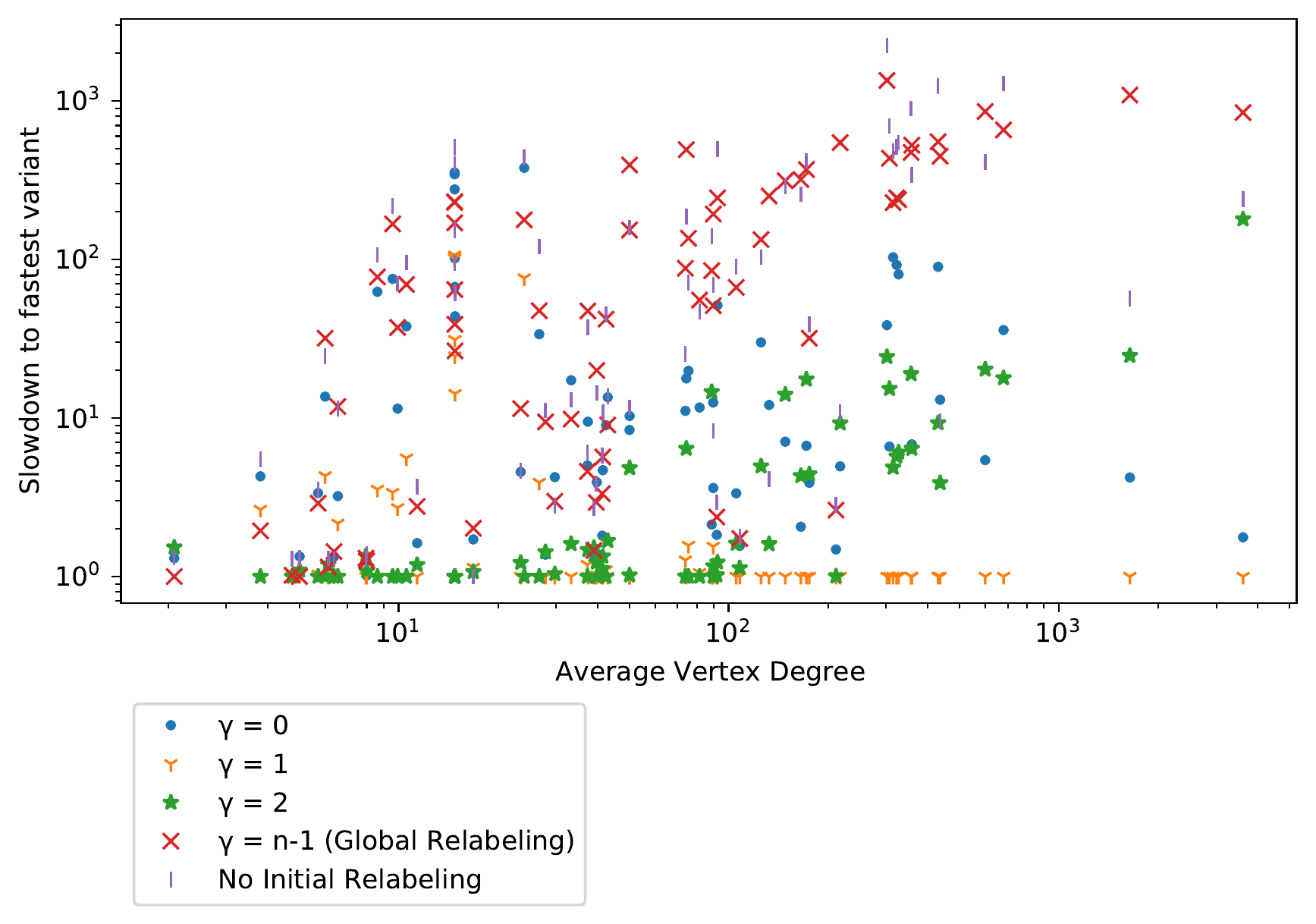}
       \caption{\label{fig:localrelabel}Effect of local relabeling depth on running time of delete operations.}
\end{figure}

In order to examine the effects of local relabeling with different values of
relabeling depth $\gamma$, we run experiments using all static graph instances
(Graph Family A and Graph Family B) from Table~\ref{table:graphs}, in which we
delete $1000$ random edges in random order. We report the total time spent
executing delete operations. We compare a total of $5$ variants, one that does
not run initial relabeling, three variants with relabeling depth $\gamma =
0,1,2$ and one variant which performs global relabeling in the initialization
process, \ie local relabeling with depth $\gamma = (n-1)$. Local relabeling with
$\gamma = 0$ is very similar to no relabeling, however the distance value of
non-sink vertices are set to $(\gamma + 1) = 1$ and not to $0$. 

In Figure~\ref{fig:localrelabel} we report the slowdown to the fastest variant
for all static graph instances from Table~\ref{table:graphs}. The x-axis shows
the average vertex degree for the instances. On most instances, the fastest
variant is local relabeling with $\gamma = 1$. Depending on the graph instance,
this variant spends $25-90\%$ of the deletion time in the initialization
(including initial relabeling). An increase in labeling depth increases the
initialization running time, but decreases the subsequent algorithm running
time. Thus we aim to find a labeling depth value that maintains some balance
between initial labeling and the subsequent algorithm execution. On some
instances, it is outperformed by local relabeling with $\gamma = 2$, which is
slower by a factor of $3-10$x on most instances, with $90-99\%$ of the total
running time spent in the initialization of the algorithm. We can see that in
instances with a higher average degree, local relabeling with $\gamma = 1$
performs better. This is an expected result, as the larger local relabeling is
more expensive in higher-average-degree graphs, as the $2$-neighborhood of a
vertex is much larger. Local relabeling with $\gamma = 2$ spends $90-99\%$ of
the total running time in initialization and initial relabeling. The same effect
is even more pronounced for the variant which performs global relabeling in
initialization. On vertices with a low average degree, we can perform global
relabeling in reasonable time, which makes the variant competitive with the
local relabeling variants. However, in high average degree instances, the
excessive running time of a global relabeling step causes the variant to have
slowdowns of up to $1000$x compared to the fastest variant. On all instances,
the vast majority of running time is spent in initialization including initial
global relabeling.

One graph family where local relabeling with $\gamma = 1$ performs badly are the
graph instances based on \texttt{auto}~\cite{karypis1998fast}, a 3D finite
element mesh graph. These graphs are rather sparse (average degree $~15$) and
planar. On these graphs, the value of the minimum cut divided by the average
degree is very large, as they do not contain any vertices of degree $1,2,3$.
Thus, the variants which perform only minor local relabeling do not guide the
flow enough and therefore the push-relabel algorithm takes a long time. On most other instances
in our test set, local relabeling with $\gamma = 1$ is enough to guide at least
$\lambda$ flow to the sink quickly.

Local relabeling with a relabeling depth $\gamma = 0$ (\ie we set the distance
of the sink to $0$, the source to $n$ and all other vertices to $1$) has a slowdown factor of
$10-100$x with only $1-10\%$ of the running time spent in the initialization.
The slowdown factor is generally increasing for larger values of the minimum cut
$\lambda$ and average degree, which indicates that ``the lack of guidance
towards the sink'' causes the algorithm to send flow to regions of the graph
that are far away from the source. For graphs with large minimum cut value
$\lambda$, the algorithm does not terminate early and needs to perform
a significant amount of push and relabel steps. In variants that perform more
relabeling at initialization, the flow is guided towards the sink by the
distance labels and the termination trigger is reached faster. The variant which
does not include any relabeling in the initialization phase has similar issues
with an even larger slowdown factor of $10-2000$x, as even flow that is already
incident to the sink does not necessarily flow straight to the sink.

On most instances, local relabeling with depth $\gamma=1$ performed best, as it
helps guide the flow towards the sink with additional work (compared to no
relabeling) only equal to the degree of the sink. While performing more
relabeling can increase this guidance even further, it comes with a trade-off in
additional time spent in the initialization. Note that this is not a general
observation for the push-relabel algorithm and can only be applied to our
application, in which the push-relabel algorithm is terminated early as soon as
$\lambda$ units of flow reach the sink vertex. Based on these experiments, we
use local relabeling with $\gamma = 1$ for edge deletions in all following
experiments.

\subsection{Dynamic Graphs}
\label{ss:exp-dynamic}

\begin{figure}[t!]
       \centering
       \includegraphics[width=\textwidth]{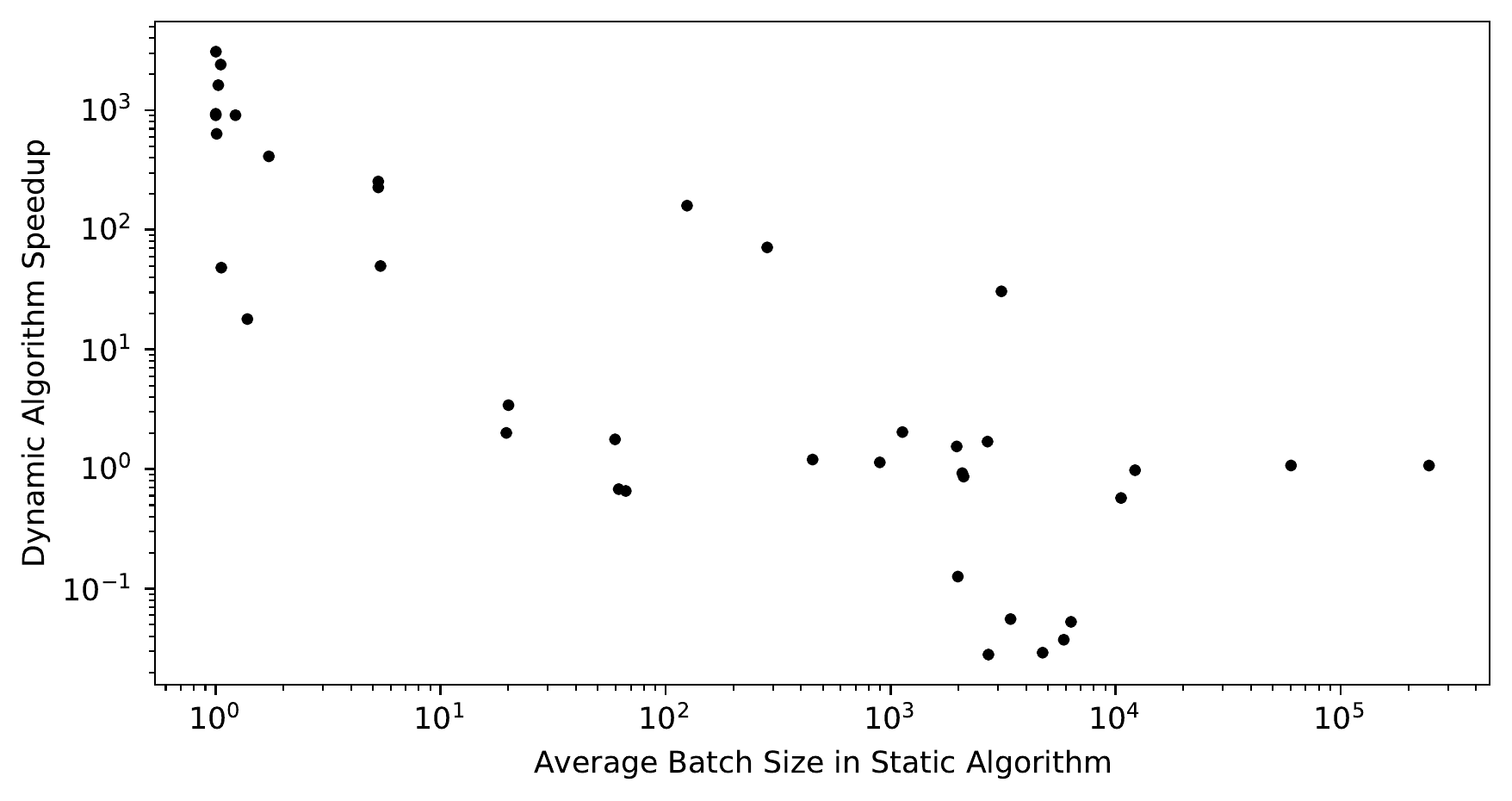}
       \caption{\label{fig:dyng}Speedup of Dynamic Algorithm.}
\end{figure}

Figure~\ref{fig:dyng} shows experimental results on the dynamic graph instances
from Graph Family C in  Table~\ref{table:graphs}. These graph instances are
mostly incremental with some being fully dynamic and most instances have
multiple connected components, \ie a minimum cut value $\lambda = 0$, even after
all insertions. On these incremental graphs with multiple connected components, our algorithm
behaves similar to a simple union-find based connected components algorithm that
for edge insertion checks whether the incident vertices already belong to the same
connected component and merges their connected components if they are different.

In this section we compare our dynamic minimum cut algorithm to
the static algorithm of Nagamochi~\etal\cite{nagamochi1994implementing}, which
has been shown to be one of the fastest sequential algorithms for the minimum
cut problem~\cite{Chekuri:1997:ESM:314161.314315,henzinger2018practical}. The
static algorithm performs the updates batch-wise, \ie the static algorithm is
not called inbetween multiple edge updates with equal timestamp. In
Figure~\ref{fig:dyng}, we show the dynamic speedup in comparison to the average
batch size. As expected, there is a large speedup factor of up to $1000$x for
graphs with small batch sizes; and the speedup decreases for increasing batch
sizes. The family of instances in which the dynamic algorithm is outperformed by
the static algorithm is the \texttt{insecta-ant-colony} graph
family~\cite{mersch2013tracking}. These graphs have a very high minimum cut
value and fewer batches than changes in the minimum cut value. Therefore, the
dynamic algorithm which updates on every edge insertion needs to recompute the
minimum cut cactus more often than the static algorithm is run and, thus, takes a longer time.

As these dynamic instances do not have sufficient diversity, we also perform experiments on static graphs in graph family B in which a subset of edges is inserted or removed dynamically. We report on this experiment in the following section.

\subsection{Random Insertions and Deletions from Static Graphs}
\label{ss:exp-static}

\begin{figure}[t!]
       \centering
       \includegraphics[width=\textwidth]{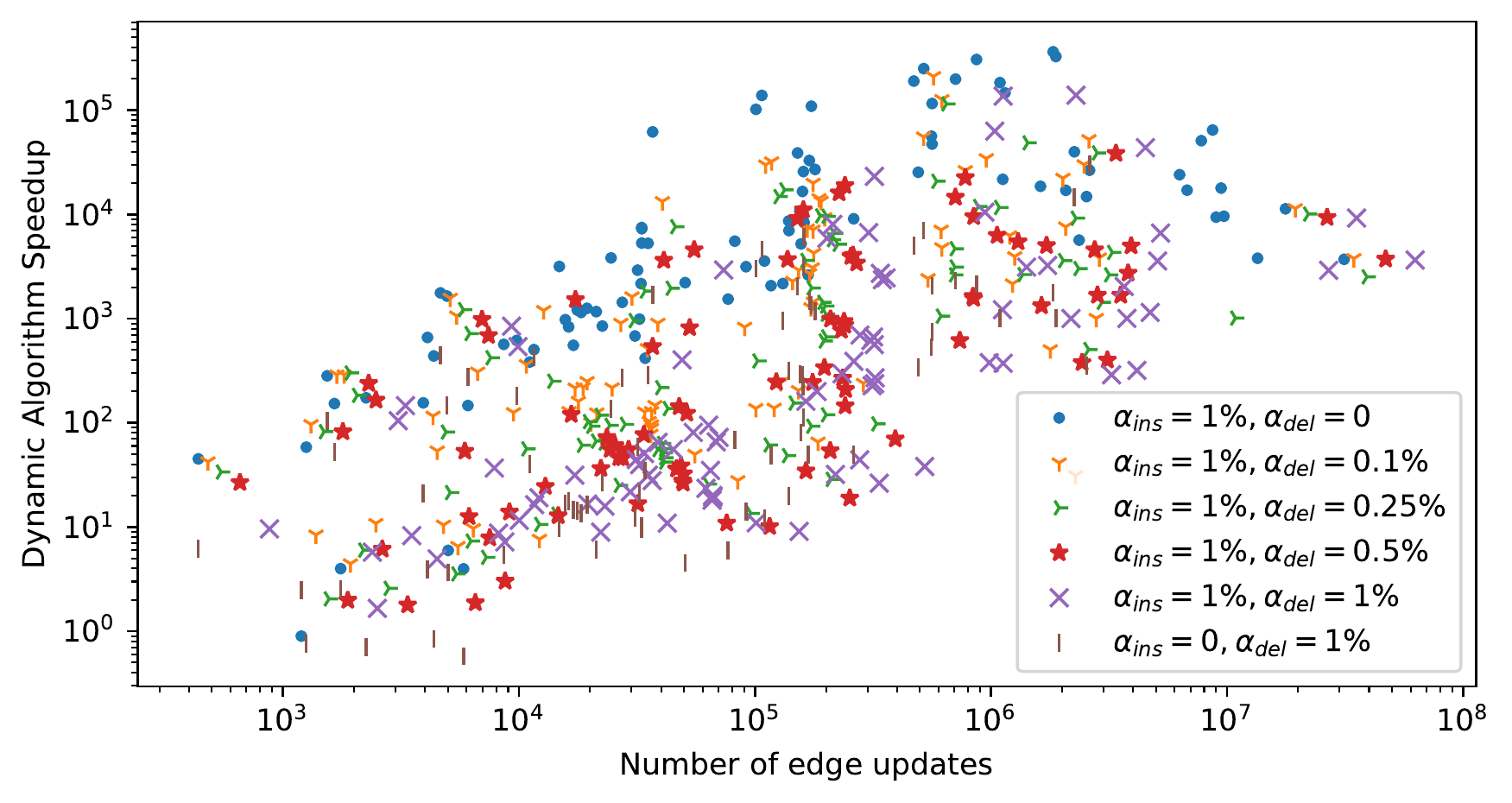}
       \caption{\label{fig:static}Speedup of Dynamic Algorithm on Random Insertions and Deletions from Static Graphs.}
\end{figure}

Figure~\ref{fig:static} shows results for dynamic edge insertions and deletions
from all graphs in Graph Family A and B from Table~\ref{table:graphs}. These
graphs are static, we create a dynamic problem from graph $G=(V,E,c)$ as
follows: let $\alpha_{ins} \in (0,1)$ and $\alpha_{del} \in (0,1)$ with
$\alpha_{ins} + \alpha_{del} < 1$ be the edge insertion and deletion rate. We
randomly select edge lists $E_{ins}$ and $E_{del}$ with $|E_{ins}| =
\alpha_{ins} \cdot |E|$, $|E_{del}| = \alpha_{del} \cdot |E|$ and $E_{ins} \cap
E_{del} = \emptyset$. For every vertex $v \in V$, we make sure that at least one
edge incident to $v$ is neither in $E_{ins}$ nor in $E_{del}$, so that the
minimum degree of $(V,E \backslash (E_{ins} \cap E_{del}), c)$ is strictly
greater than $0$ at any point in the update sequence.

We initialize the graph as $(V,E \backslash E_{ins}, c)$ and create a sequence of edge updates $E_u$ by concatenating $E_{ins}$ and $E_{del}$ and randomly shuffling the combined list. Then we perform edge updates one after another and compute the minimum cut - either statically using the algorithm of Nagamochi~\etal\cite{nagamochi1994implementing} or by performing an update in the dynamic algorithm - after every update. We report the total running time of either variant and give the speedup of the dynamic algorithm over the static algorithm as a function of the number of edge updates performed. For each graph we create problems with $\alpha_{ins} = 1\%$ and $\alpha_{del} \in \{0,0.1\%,0.25\%,0.5\%,1\%\}$; and additionally 
a decremental problem with $\alpha_{ins} = 0$ and $\alpha_{del} = 1\%$. We set the timeout for the static algorithm to $1$ hour, if the algorithm does not finish before timeout, we approximate the total running time of the static algorithm by performing $100$ or $1000$ updates in batch.

Dynamic edge insertions are generally much faster than edge deletions, as most real-world graphs have large sets that are not separated by any global minimum cut. When inserting an edge where both incident vertices are in the same set in $\mathcal{C}$, the edge insertion only requires two array accesses; if they are in different sets, it requires a breadth-first search on the relatively small cactus graph $\mathcal{C}$ and only if there are no minimum cuts remaining, an edge insertion requires a recomputation. In contrast to that, every edge deletion requires solving of a flow problem and therefore takes significantly more time in average. Therefore, the average speedup is larger on problems with a higher rate of edge insertions.

Generally, the speedup of the dynamic algorithm increases with larger problems and more edge updates. For larger graphs with $\geq 10^6$ edge updates, the average speedup is more than four orders of magnitude for instances with $\alpha_{del} = 0$ and still more than two orders of magnitude for large instances when $\alpha_{del} = \alpha_{ins} = 1\%$. Note that in this experiment, the number of edge updates is a function of the number of edges, thus instances with more updates directly correspond to graphs with more edges.

For decremental instances with $\alpha_{ins} = 0$, the speedup is generally lower, but still reaches multiple orders of magnitude in larger instances.

\subsubsection{Most Balanced Minimum Cut}

Henzinger~\etal\cite{henzinger2020finding} show that given the cactus graph $\mathcal{C}$ we can compute the most balanced minimum cut, \ie the minimum cut which has the highest number of vertices in the smaller partition, in $\Oh{n^*}$ time. In our algorithm for the dynamic minimum cut problem we also compute a cactus graph of minimum cuts, however this cactus graph does not necessarily contain all minimum cuts in $G$, as we do not introduce new minimum cuts added by edge deletions. 

We use the algorithm of Henzinger~\etal\cite{henzinger2020finding} to find the
most balanced minimum cut for all instances of Graph Family B every $1000$ edge
updates and compare it to the most balanced minimum cut found by our algorithm.
In instances that are not just decremental, in $97.3\%$ of all cases where there
is a nontrivial minimum cut (\ie smaller side contains multiple vertices), both
algorithms give the same result, \ie our algorithm can almost always output the most balanced
minimum cut. In the instances that are purely decremental, \ie $|E_{ins}| = 0$, we
only find the most balanced minimum cut in $25.4\%$ of cases where there is a
non-trivial minimum cut. This is the case because an increase of the minimum cut
prompts a full recomputation of a cactus graph that represents all (potentially
many) minimum cuts, thus also the most balanced minimum cut. Only if this cut in
particular is affected by an edge update, the dynamic algorithm ``loses'' it. In
the purely decremental case, the minimum cut value only decreases. Thus, the
dynamic algorithm only knows one or a few minimum cuts. All cuts that reach the
same value $\lambda$ in later edge deletions are not in $\mathcal{C}$, as we do
not add cuts of the same value to it. As these decremental instances do not
have any edge insertions that can increase the value of these cuts, there is
eventually a large set of minimum cuts of which the algorithm only knows a few. 
If maintaining a balanced minimum cut is a requirement, this can easily be
achieved by occasionally recomputing the entire cactus graph $\mathcal{C}$ from
scratch.

\subsection{Worst-case instances}
\label{ss:exp-worstcase}

\begin{figure}[t!]
       \centering
       \includegraphics[width=\textwidth]{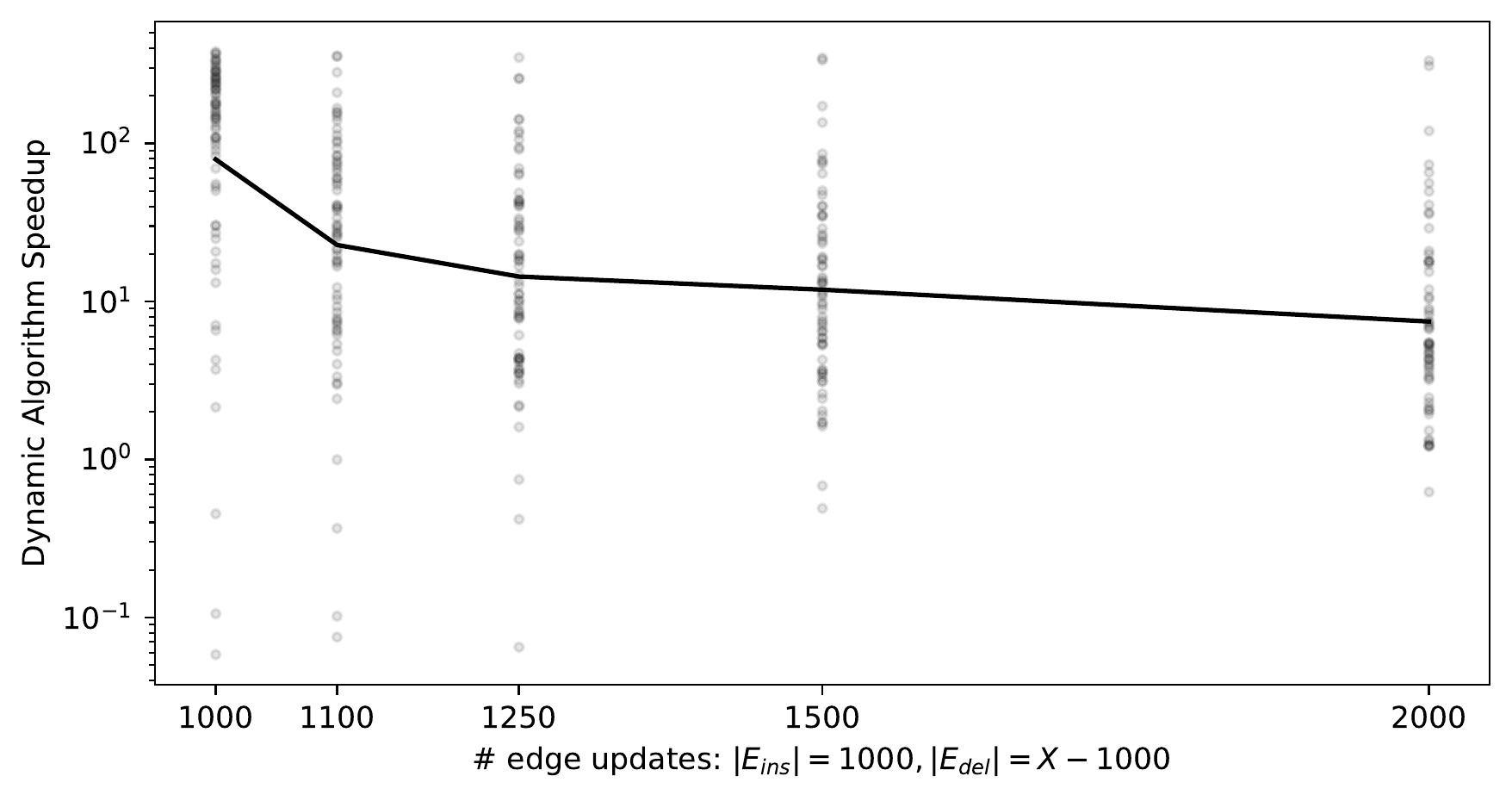}
       \caption{\label{fig:worstcase}Speedup of Dynamic Algorithm on Worst-case Insertions and Deletions from Static Graphs.}
\end{figure}

On random edge insertions, there is a high chance that the vertices incident to
the newly inserted edge were not separated by a minimum cut and therefore
require no update of the cactus graph $\mathcal{C}$. In this experiment we aim
to generate instances that aim to maximize the work performed by the dynamic
algorithm. We initialize the graph as $G=(V,E,c)$ and add random unit-weight
edges $e = (u,v)$ where $\Pi(u) \neq \Pi(v)$ for every newly added edge. Then we
randomly select $|E_{ins}| = 1000$ edges to add so that for each such edge $(u,v)$, $\Pi(u)
\neq \Pi(v)$ before inserting $(u,v)$, and select a subset $E_{del}
\subseteq E_{ins}$ to delete. For each graph we create $5$ problems, with
$|E_{del}| \in \{0,100,250,500,1000\}$. We randomly shuffle the edge updates
while making sure that an edge deletion is only performed after the respective
edge has been added to the graph, but still interspersing edge insertions and
deletions to create true worst-case instances for the dynamic algorithm, as each
edge deletion or insertion affects one or multiple minimum cuts in the graph.

Figure~\ref{fig:worstcase} shows the results of this experiment. Each low-alpha dot shows the speedup of the dynamic algorithm on a single problem, the black line gives the geometric mean speedup. As indicated in previous experiments, we can see that the average speedup decreases when the ratio of deletions is increased. However, even on these worst-case instances, the mean speedup factor is still $7.46$x for $|E_{ins}| = |E_{del}| = 1000$ up to $79.2$x for the purely incremental instances on instances where both algorithms finished before timeout at one hour. Similar to previous experiments, the speedup factor increases with the graph size.

On these problem instances we can see interesting effects. Especially in instances with $|E_{del}| = 500$ we can see many instances where the minimum cut fluctuates between two different values in more than half of all edge updates. As the larger of the values usually has a large cactus graph $\mathcal{C}$, this would result in expensive recomputation on almost every update. However, using the cactus caching technique detailed in Section~\ref{ss:cache} we can save this overhead and simply reuse the almost unchanged previous cactus graph. In some cases, this reduces the number of calls to the algorithm of Henzinger~\etal\cite{henzinger2020finding} by more than a factor of $10$.

We also find some instances where the static graph has few minimum cuts, but
there is a large set of cuts slightly larger than lambda. One such example are planar graphs
derived from Delaunay triangulation~\cite{lee1980two} that have a few vertices
of minimal degree near the edges of the triangulated object, but a large number
of vertices with a slightly larger degree. If we now add edges to increase the
degree of the minimum-degree vertices, the resulting graph has a huge number of
minimum cuts and computing all minimum cuts is significantly more expensive than
computing just a single minimum cut. In these instances the dynamic algorithm is
actually slower than rerunning the static algorithm on every edge update. The
dynamic algorithm is slower than the static algorithm in $3.9\%$ of the worst-case instances.

\section{Conclusion}

In this work, we presented the first implementation of a fully-dynamic algorithm that maintains the minimum cut of a graph under both edge insertions and deletions. Our algorithm combines ideas from the theoretical foundation with efficient and fine-tuned implementations to give an algorithm that outperforms static approaches by up to five orders of magnitude on large graphs. In our experiments, we show the performance of our algorithm on a wide variety of graph instances.

Future work includes maintaining all global minimum cuts also under edge deletions and employing shared-memory or distributed parallelism to further increase the performance of our algorithm.

\ifVLDB
\else
\bibliographystyle{plainurl}
\fi{}
\bibliography{paper}


\end{document}